\documentclass[a4paper,11pt]{article}

\usepackage{amsmath}
\usepackage{amssymb}
\usepackage{amsthm}
\usepackage{yhmath}
\usepackage{graphicx}
\usepackage{enumerate}
\usepackage{cite}
\usepackage{fullpage}
\usepackage[usenames]{color}
\usepackage[nodayofweek]{datetime}

\allowdisplaybreaks

\newcommand{\Real}{\ensuremath{\mathbb{R}}}
\newcommand{\Plane}{\ensuremath{\mathbb{R}^2}}

\newcommand{\Poly}{\ensuremath{\mathcal{P}}}
\newcommand{\SPM}{\ensuremath{\mathsf{SPM}}}
\newcommand{\Vis}{\ensuremath{\mathsf{VR}}}
\newcommand{\Bd}{\ensuremath{\mathcal{B}}}
\newcommand{\plf}{\ensuremath{\mathrm{len}}}
\newcommand{\start}{\ensuremath{\alpha}}
\newcommand{\final}{\ensuremath{\omega}}
\newcommand{\Vpair}{\ensuremath{\mathcal{V}}}

\newcommand{\bd}{\ensuremath{\partial}}

\newcommand{\intr}{\ensuremath{\mathrm{int}}}

\newcommand{\seg}{\overline}

\newcommand{\diam}{\mathrm{diam}}
\newcommand{\dist}{\mathrm{d}}

\newtheoremstyle{mytheorem}{3pt}{3pt}{\slshape}{}{\bfseries}{}{.5em}{}
\theoremstyle{mytheorem}
\newtheorem{lemma}{Lemma}
\newtheorem{theorem}{Theorem}

\newtheorem{corollary}{Corollary}

\newtheorem{claim}{Claim}
\theoremstyle{definition}

\makeatletter
\newbox\ProofSym
\setbox\ProofSym=\hbox{%
\unitlength=0.18ex%
\begin{picture}(10,10)
\put(0,0){\framebox(9,9){}} \put(0,3){\framebox(6,6){}}
\end{picture}}
\renewenvironment{proof}[1][Proof.]{\O@proof{#1}}{\O@endproof}
\def\O@proof#1{\trivlist
   \@topsep\z@\@topsepadd\smallskipamount%
   \@ifstar{\item[]}{\item[\hskip\labelsep\it #1 ]}}
\def\O@endproof{\hfill\copy\ProofSym\linebreak[3mm]\endtrivlist}
\makeatother

\def\denseitems{
    \itemsep1pt plus1pt minus1pt
    \parsep0pt plus0pt
    \parskip0pt\topsep0pt}
\title{The Geodesic Diameter of Polygonal Domains%
\thanks{%
A preliminary version of this paper was presented at the 18th Annual European Symposium on Algorithms (ESA 2010).
Work by S.W. Bae was supported by National Research Foundation of
Korea (NRF) grant
funded by the Korea government (MEST) (No. 2010-0005974).
Work by Y. Okamoto was supported by Global COE Program
``Computationism as a Foundation for the Sciences'' and Grant-in-Aid
for Scientific Research from Ministry of Education, Science and
Culture, Japan, and Japan Society for the Promotion of Science.
}
}

\author{%
Sang Won Bae\thanks{%
Department of Computer Science, Kyonggi University, Suwon, Korea.
Email: \texttt{swbae@kgu.ac.kr}
}
\and
Matias Korman\thanks{%
Computer Science Department, Universit\'e Libre de Bruxelles (ULB), Belgium.
Email: \texttt{mkormanc@ulb.ac.be}
}
\and %
Yoshio Okamoto\thanks{%
Center for Graduate Education Initiative,
Japan Advanced Institute of Science and Technology, Nomi, Japan.
Email: \texttt{okamotoy@jaist.ac.jp} } }

\date{%
\today\quad\currenttime
}

\begin{document}

\maketitle

\begin{abstract}
 This paper studies the geodesic diameter of polygonal domains
 having $h$ holes and $n$ corners.
 For simple polygons (i.e., $h=0$),
 the geodesic diameter is determined by a pair of corners of a given polygon
 and can be computed in linear time, as known by Hershberger and Suri.
 For general polygonal domains with $h \geq 1$, however,
 no algorithm for computing the geodesic diameter was known prior to this paper.
 In this paper, we present the first algorithms that compute the geodesic diameter
 of a given polygonal domain in worst-case time $O(n^{7.73})$ or $O(n^7 (\log n + h))$.
 The main difficulty unlike the simple polygon case relies on
 the following observation revealed in this paper:
 two interior points can determine the geodesic diameter and in that case
 there exist at least five distinct shortest paths between the two.
\end{abstract}

\section{Introduction} \label{sec:intro}
A \emph{polygonal domain} $\Poly$ with $h$ holes and $n$ corners is a connected and closed subset of $\Plane$ having $h$ holes whose boundary consists of $h+1$ simple closed polygonal chains of $n$ total line segments.
Given a polygonal domain $\Poly$, the geodesic distance $\dist(p,q)$ between two points $p$ and $q$ of $\Poly$ is defined as the length of a shortest path that connects $p$ and $q$
and stays within $\Poly$.

This paper addresses the geodesic diameter problem in polygonal domains having one or more holes.
The geodesic diameter $\diam(\Poly)$ of domain $\Poly$ is defined as the largest possible geodesic distance between any two points of $\Poly$, that is, $\diam(\Poly) = \max_{s,t \in \Poly} \dist(s,t)$.



For simple polygons (i.e., domains with no hole), the geodesic diameter has been extensively studied. Chazelle~\cite{c-tpca-82} provided the first $O(n^2)$-time algorithm computing the geodesic diameter of a simple polygon.
Afterwards, Suri~\cite{s-agfnpsp-87} presented an $O(n\log n)$-time algorithm that solves the all-geodesic-farthest neighbors problem, computing the farthest neighbor of every corner and thus finding the geodesic diameter. At last, Hershberger and Suri~\cite{hs-msspm-97} showed that the diameter can be computed in linear time using fast matrix search techniques.

On the other hand, the geodesic diameter of a domain having one or more holes is less understood. Mitchell~\cite{m-spn-04} has posed an open problem asking an algorithm for computing the geodesic diameter of a polygonal domain.
However, even for the corner-to-corner diameter $\max_{u,v\in V} \dist(u,v)$,
where $V$ denotes the set of corners of $\Poly$,
we know nothing better than a brute-force algorithm that takes $O(n^2\log n)$ time,
checking all the geodesic distances between every pair of corners.\footnote{%
Personal communication with Joseph S. B. Mitchell.}
Prior to our results,
there was no known algorithm for computing the geodesic diameter in domains with holes.
We should also mention that Koivisto and Polishchuk~\cite{kp-gdpd-10} had
claimed an improved algorithm after a preliminary report of our work~\cite{bko-gdpd-10},
but it was shown to be a failed trial through conversations with the authors.\footnote{%
Personal communication with Valentin Polishchuk.}

This fairly wide gap between simple polygons and polygonal
domains with holes is seemingly due to the uniqueness of the
shortest path between any two points.
When a domain $\Poly$ has no hole, it is well known that
there is a unique shortest path between any two points~\cite{gh-ospqsp-89}. 
Using this uniqueness, one can show that the diameter $\diam(\Poly)$
is realized by a pair of corners~\cite{hs-msspm-97,s-agfnpsp-87}.
For general polygonal domains, however, this is not the case.
In this paper, we exhibit several examples where the diameters are realized by
non-corner points on $\bd \Poly$ or even by interior points
of $\Poly$. See \figurename~\ref{fig:examples1}. 
Such examples were constructed based on the multiplicity
of shortest paths and, to our best knowledge, never known prior to this work.
This observation also shows an immediate difficulty in devising any exhaustive algorithm
since one sees no intuitive discretization of the search space.

The status of the geodesic center problem is also similar.
A point in $\Poly$ is defined as a \emph{geodesic center} if it minimizes the maximum geodesic distance from it to any other point of $\Poly$.
Asano and Toussaint~\cite{at-cgcsp-85} introduced
the first $O(n^4\log n)$-time algorithm for computing the geodesic center
of a simple polygon (i.e., when $h=0$), and
Pollack, Sharir and Rote~\cite{psr-cgcsp-89} improved it to $O(n\log n)$ time.
As with the diameter problem, there is no known algorithm for domains with holes.
See O'Rourke and Suri~\cite{os-p-04} and Mitchell~\cite{m-spn-04}
for more references on the geodesic diameter/center problem.

\begin{figure}[t]
\begin{center}
  \includegraphics[width=.98\textwidth]{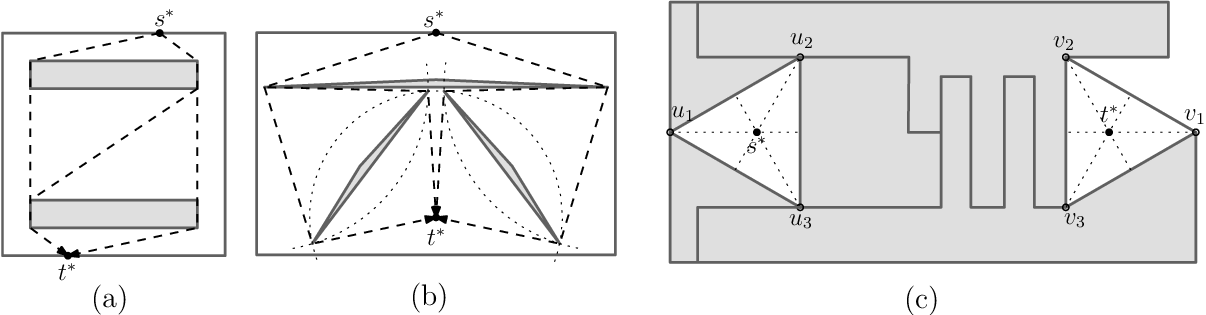}
\end{center}
\caption{\small Three polygonal domains where the geodesic diameter
 is determined by a pair $(s^*, t^*)$ of non-corner points;
 Gray-shaded regions depict the interior of the holes and dark gray segments depict the boundary $\bd\Poly$.
 Recall that $\Poly$, as a set, contains its boundary $\bd\Poly$.
 (a) Both $s^*$ and $t^*$ lie on $\bd \Poly$.
 There are three shortest paths between $s^*$ and $t^*$.
 In this domain, there are two (symmetric) diametral pairs (only one is depicted for clarity).
 (b) $s^* \in \bd \Poly$ and $t^* \in \intr\Poly$.
 Three triangular holes are placed in a symmetric way, obtaining four shortest paths between $s^*$ and $t^*$.
 (c) Both $s^*$ and $t^*$ lie in the interior $\intr\Poly$.
 Here, the five holes are packed like jigsaw puzzle pieces,
 forming narrow corridors (dark gray paths) and two empty, regular triangles.
 Observe that  $\dist(u_1,v_1)=\dist(u_1, v_2)= \dist(u_2, v_2)=\dist(u_2,v_3)$
 $=\dist(u_3,v_3)=\dist(u_3,v_1)$.
 The points $s^*$ and $t^*$ lie at the centers of the triangles formed
 by the $u_i$ and the $v_i$, respectively.
 There are six shortest paths between $s^*$ and $t^*$.
 }
\label{fig:examples1}
\end{figure}

Since the geodesic diameter/center of a simple polygon is determined by its corners,
one can exploit the \emph{geodesic farthest-site Voronoi diagram} of the set $V$ of corners
to compute the diameter/center,
which can be built in $O(n \log n)$ time~\cite{afw-fsgvd-93}.
Recently, Bae and Chwa~\cite{bc-gfvdpdh-09} presented an $O(nk \log^3 (n+k))$-time algorithm
for computing the geodesic farthest-site Voronoi diagram of $k$ sites in polygonal domains with holes.
This result can be used to compute the geodesic diameter $\max_{p,q\in S} \dist(p,q)$ of
a \emph{finite} set $S$ of points in $\Poly$. However, this approach cannot be directly used for computing $\diam(\Poly)$ without any characterization of the diameter. Moreover, when $S=V$, this approach is no better than the brute-force $O(n^2 \log n)$-time algorithm for computing the corner-to-corner diameter $\max_{u,v \in V}\dist(u,v)$.

In this paper, we present the first algorithms that compute
the geodesic diameter of a given polygonal domain in
$O(n^{7.73})$ or $O(n^7(\log n + h))$ time in the worst case.
Our new geometric results underlying the algorithms show that
the existence of any diametral pair consisting of non-corner points implies
multiple shortest paths between the pair;
among other results, we show that 
\emph{if $(s,t)$ is a diametral pair and both $s$ and $t$ lie in the interior of $\Poly$,
then there are at least five shortest paths between $s$ and $t$.}

Some analogies between polygonal domains and convex polytopes in $\Real^3$ can be seen.
O'Rourke and Schevon~\cite{os-cgd3p-89} proved 
that if the geodesic diameter on a convex $3$-polytope is realized by two non-corner points,
then at least five shortest paths exist between the two;
see also Zalgaller~\cite{z-ipt-07} for simpler arguments.
Based on this observation, they presented an $O(n^{14} \log n)$-time algorithm for computing
the geodesic diameter on a convex $3$-polytope.
Afterwards, the time bound was improved to $O(n^8\log n)$ by Agarwal et al.~\cite{aaos-supa-97}
and recently to $O(n^7 \log n)$ by Cook IV and Wenk~\cite{cw-sppps-09}.


The rest of the paper is organized as follows:
After introducing preliminary definitions and concepts in Section~\ref{sec:pre},
we investigate local maxima of the lower envelope of convex functions in Section~\ref{sec:convex}, resulting in Theorem~\ref{theorem:convex_max}.
Section~\ref{sec:charac} extensively exploits the intermediate result
to show lower bounds on the number of shortest paths
between a diametral pair for every possible case, and then
Section~\ref{sec:algorithm} describes our algorithms for the geodesic diameter.
We finally concludes the paper with summary, some remarks, and open issues
in Section~\ref{sec:conclusion}.
Also, we exhibit several interesting examples that cover all possible combinatorial cases
in Appendix~\ref{sec:moreexample}.
We hope that the readers will enjoy them.

\section{Preliminaries} \label{sec:pre}
Throughout the paper, we frequently use several topological concepts such as
open and closed subsets, neighborhoods, and the boundary $\bd A$ and the interior $\intr A$
of a set $A$;
unless stated otherwise, all of them are supposed to be derived with respect to
the standard topology on $\Real^d$ with the Euclidean norm $\|\cdot\|$ for fixed $d\geq 1$. We also denote the straight line segment joining two points $a, b$ by $\seg{ab}$.

A \emph{polygonal domain} $\Poly$ with $h$ holes and $n$ corners\footnote{%
We reserve the term ``vertex'' for 0-dimensional faces of subdivisions of a certain space.}
is a connected and closed subset of $\Plane$ with $h$ holes whose boundary $\bd \Poly$
consists of $h+1$ simple closed polygonal chains of $n$ total line segments.
The boundary $\bd \Poly$ of a polygonal domain $\Poly$ is regarded as a series of
\emph{obstacles} so that any feasible path in $\Poly$ is not allowed to cross
$\bd \Poly$.
The \emph{geodesic distance} $\dist(p,q)$ between any two points $p,q$ in a polygonal domain $\Poly$ is defined as the length of a shortest feasible path between $p$ and $q$, where the \emph{length} of a path is the sum of the Euclidean lengths of its segments.
It is well known from earlier work~\cite{m-spaop-96} that there always exists a shortest feasible path between any two points $p, q \in \Poly$,
and thus the geodesic distance function $\dist(\cdot, \cdot)$ is well defined.
The geodesic diameter $\diam(\Poly)$ of a polygonal domain $\Poly$ is defined as the largest geodesic distance between any two points of $\Poly$, that is,
 \[\diam(\Poly) = \max_{s,t \in \Poly} \dist(s,t).\]
A pair $(s,t)$ of points in $\Poly$ that realizes the geodesic diameter $\diam(\Poly)$
is called a \emph{diametral pair}.

\paragraph{Shortest path map.}
Let $V$ be the set of all corners of $\Poly$
and $\pi(s,t)$ be a shortest path between $s\in\Poly$ and $t\in\Poly$.
Such a path $\pi(s,t)$ is represented as a sequence
$\pi(s,t) = (s, v_1, \ldots, v_k, t)$ for some $v_1, \ldots, v_k\in V$;
that is, a polygonal chain through a sequence of corners~\cite{m-spaop-96}.
Note that we can have $k=0$ when $\dist(s,t)=\|s-t\|$.
If two paths (with possibly different endpoints) induce the same sequence of corners
$(v_1, \ldots, v_k)$,
then they are said to have the same \emph{combinatorial structure}.

The \emph{shortest path map} $\SPM(s)$ for a fixed $s\in\Poly$ is a
decomposition of $\Poly$ into cells such that every point in a common cell can be reached
from $s$ by shortest paths of the same combinatorial structure. Each
cell $\sigma_s(v)$ of $\SPM(s)$ is associated with a corner $v\in V$ which is the last corner of $\pi(s,t)$ for any $t$ in the cell $\sigma_s(v)$.
We also define the cell $\sigma_s(s)$ as the set of points $t\in\Poly$
such that $\pi(s,t)$ passes through no corner of $\Poly$, so $\pi(s,t) = \seg{st}$.
Each edge of $\SPM(s)$ is an arc on the boundary of two incident cells
$\sigma_s(v_1)$ and $\sigma_s(v_2)$ determined by two
corners $v_1, v_2 \in V\cup\{s\}$.
Similarly, each vertex of $\SPM(s)$ is determined by at least three distinct corners
$v_1,v_2,v_3 \in V \cup \{s\}$.

Note that, for fixed $s\in\Poly$,
a point farthest apart from $s$ lies
at either (1) a vertex of $\SPM(s)$, (2) an intersection between the boundary $\bd\Poly$
and an edge of $\SPM(s)$, or (3) a corner in $V$.
The shortest path map $\SPM(s)$ 
has $O(n)$ total number of cells, edges, and vertices
and
can be computed in $O(n\log n)$ time using $O(n\log n)$ working space~\cite{hs-oaespp-99}.
For more details on shortest path maps, see~\cite{m-spaop-96, hs-oaespp-99, m-spn-04}.

\paragraph{Path-length function.}
If $\pi(s,t) \neq \seg{st}$, then there are two corners $u,v\in V$ such that
$u$ and $v$ are the first and last corners along $\pi(s,t)$ from $s$ to $t$, respectively.
Here, the path $\pi(s,t)$ is formed as the union of
$\seg{su}$, $\seg{vt}$ and a shortest path $\pi(u,v)$ from $u$ to $v$.
Note that $u$ and $v$ are not necessarily distinct.
In order to realize such a path, we assert that $s$ is visible from $u$ and
$t$ is visible from $v$. That is, $s\in \Vis(u)$ and $t\in \Vis(v)$,
where $\Vis(p)$ for any $p\in\Poly$ is defined to be the set of all points $q\in\Poly$
such that $\seg{pq}\subset \Poly$, also called the \emph{visibility region} of $p\in\Poly$.

We now define the \emph{path-length function} $\plf_{u,v}\colon \Vis(u)\times \Vis(v) \to \Real$ for any fixed pair of corners $u,v\in V$ to be
 \[\plf_{u,v}(s,t) := \|s-u\| + \dist(u,v) + \|v-t\|.\]
That is, $\plf_{u,v}(s,t)$ represents the length of paths from $s$ to $t$ that have a common combinatorial structure; going straight from $s$ to $u$, following a shortest path from $u$ to $v$,
and going straight to $t$.
Also, unless $\dist(s,t) = \|s-t\|$ (equivalently, $s\in \Vis(t)$), the geodesic distance $\dist(s,t)$ can be expressed as the pointwise minimum of some path-length functions:

\[ \dist(s,t) = \min_{u \in \Vis(s),~ v\in \Vis(t)} \plf_{u,v}(s,t).\]

Consequently, we have two possibilities for a diametral pair $(s^*,t^*)$; either we have $\dist(s^*,t^*) = \|s^*-t^*\|$ or the pair $(s^*,t^*)$ is a local maximum of the lower envelope of several path-length functions.
In the following, we will mainly study the latter case, since the former can be easily handled.

\section{Local Maxima of the Lower Envelope of Convex Functions} \label{sec:convex}
In this section, we give an interesting  property of the lower envelope of a family of convex functions which will afterwards be used in our geodesic diameter environment. We start with a basic observation on the intersection of hemispheres on a unit hypersphere in the $d$-dimensional space $\Real^d$. For any fixed positive integer $d$, let ${S}^{d-1}:=\{x \in \Real^d \mid \|x\| = 1\}$ be the unit hypersphere in $\Real^d$ centered at the origin. A \emph{closed} (or \emph{open}) \emph{hemisphere} on $S^{d-1}$ is defined to be the intersection of $S^{d-1}$ and a closed (open, respectively) half-space of $\Real^d$ bounded by a hyperplane that contains the origin.

We call a $k$-dimensional affine subspace of $\Real^d$ a \emph{$k$-flat}. Note that a hyperplane in $\Real^d$ is a $(d-1)$-flat and a line in $\Real^d$ is a $1$-flat. Also, the intersection of $S^{d-1}$ and a $k$-flat through the origin in $\Real^d$ is called a \emph{great $(k-1)$-sphere} on $S^{d-1}$. Note that a great $1$-sphere is called a great circle and a great $0$-sphere consists of two antipodal points.

\begin{lemma} \label{lemma:hemisphere}
 For any two positive integers $d$ and $m\leq d$, a set of any $m$ closed hemispheres on $S^{d-1}$ has a nonempty common intersection.
 Moreover, if the intersection has an empty interior relative to $S^{d-1}$,
 then it includes a great $(d-m)$-sphere on $S^{d-1}$.
\end{lemma}
\begin{proof}
We only give a proof for the second statement, which implies the first.
The case of $d=1$ is trivial, so we assume $d>1$.
Let $H_1,\ldots, H_m$ be any $m$ closed hemispheres on $S^{d-1}$, and
$h_i$ be the hyperplane through the origin in $\Real^d$ such that
$H_i$ lies in a closed half-space supported by $h_i$.
In this proof, we denote by $\wideparen{H}_i$ the open hemisphere,
defined to be $\wideparen{H}_i = H_i \setminus h_i$.
Also, let $\mathcal{H}_j := \bigcap_{1\leq i \leq j} H_i$ and
$\wideparen{\mathcal{H}}_j := \bigcap_{1\leq i \leq j} \wideparen{H}_i$.

Suppose that $\wideparen{\mathcal{H}}_m = \emptyset$.
Let $k$ be  the smallest integer such that $\wideparen{\mathcal{H}}_k = \emptyset$.
By definition, $k \geq 2$ and $\wideparen{\mathcal{H}}_{k-1} \neq \emptyset$.
Note that the intersection of any $k-1$ non-parallel hyperplanes of $\Real^d$
includes a $(d-k+1)$-flat 
and each $h_i$ contains the origin.
Hence, $\bigcap_{1\leq i\leq k-1} h_i$ includes a $(d-k+1)$-flat through the origin and thus
$\mathcal{H}_{k-1}$ includes a great $(d-k)$-sphere $G$ on $S^{d-1}$.
Since $x\in G$ implies $-x \in G$ for any $x\in S^{d-1}$,
we must have $G \subseteq h_k$, in order to have an empty intersection $\wideparen{\mathcal{H}}_k$.
This implies that $\bigcap_{1\leq i\leq k} h_i$ also includes a $(d-k+1)$-flat through the origin,
and further that $\bigcap_{1\leq i\leq m} h_i$ includes a $(d-m+1)$-flat through the origin.
We hence conclude that $\mathcal{H}_m = \bigcap_{1\leq i \leq m} H_i$ includes a great $(d-m)$-sphere
on $S^{d-1}$.
\end{proof}

Using Lemma~\ref{lemma:hemisphere} we prove the following theorem.

\begin{theorem} \label{theorem:convex_max}
 For any fixed positive integer $d$,
 let $\mathcal{F}$ be a finite family of real-valued convex functions
 defined on a convex subset $C\subseteq \Real^d$ and
 $g(x):= \min_{f \in\mathcal{F}} f(x)$ be their pointwise minimum.
 Suppose that $g$ attains a local maximum at $x^* \in C$ and
 there are exactly $m \leq d$ functions $f_1,\ldots,f_m \in\mathcal{F}$
 such that $f_i(x^*)=g(x^*)$ for each $i=1,\ldots,m$.
 Then, there exists a $(d+1-m)$-flat $\varphi \subset \Real^d$ through $x^*$
 such that $g$ is constant on $\varphi \cap U$
 for some neighborhood $U \subset \Real^d$ of $x^*$ with $U\subset C$.
\end{theorem}
\begin{proof}
First, we give a sketchy idea of our proof for the theorem. All functions $f \in \mathcal{F}$ other than $f_1,\ldots, f_m$ must satisfy $f(x)>g(x)$ in a small neighborhood of $x^*$.
In particular, the function $g$ is the lower envelope of the $m$ convex functions $f_i$
in a small neighborhood of $x^*$.
By convexity, we will show that for each $i$, there is a hemisphere $H_i$ of directions in $S^{d-1}$ in which $f_i$ does not decrease.
(Note that the sphere $S^{d-1}$ represents the space of all directions in $\Real^d$.)
This result combined with Lemma \ref{lemma:hemisphere} gives that the intersection of hemispheres will be a $(d+1-m)$-flat in which neither of the $m$ functions (nor $g$) can decrease.
Since $x^*$ is a local maximum of $g$, the only possibility is
that $g$ remains constant near $x^*$ along the flat.

A more detailed proof is given as follows.
Let $x^* \in C$ and $f_1, \ldots, f_m \in \mathcal{F}$ be as in the statement.
For each $i$, consider the sublevel set $L_i := \{ x \in C \mid f_i(x) \leq f_i(x^*) \}$.
Here, we consider two cases:
(i) $x^*$ lies in the interior of $L_i$ or (ii) on its boundary $\bd L_i$.
Note that $L_i \subseteq C$ is convex since $f_i$ is a convex function.
For the latter case (ii), there exists a supporting hyperplane $h_i$ to $L_i$ at $x^*$
since $L_i$ is convex and $x^* \in \bd L_i$.
Denote by $h_i^\oplus$ the closed half-space that is bounded by $h_i$ and does not contain $L_i$.
For the former case (i), we choose $h_i$ to be any hyperplane of $\Real^d$ through $x^*$
and $h_i^\oplus$ to be any closed half-space supported by $h_i$.
Then, we have that $f_i(x^*) \leq f_i(x)$ for any $x\in h_i^\oplus \cap C$,
regardless of the cases;
in particular for Case (i), observe that $f_i(x) = f_i(x^*)$ for any $x\in L_i$ by convexity
so that we can choose any hyperplane as $h_i$.


%


Now, we let
 \[H_i := \{x-x^* \mid x\in h_i^\oplus, \|x-x^*\|=1\}\]
be a closed hemisphere of the unit sphere $S^{d-1}$ centered at the origin.
Note 
that $f_i$ does not decrease if we move from $x^*$ locally in any direction in $H_i$.
Since $g(x^{*})=f_i(x^{*})$ for any $i\in \{1,\ldots,m\}$ and $x^*$ is a local maximum of $g$,
the intersection $\bigcap_{i=1}^{m} H_i$ has an empty interior relative to $S^{d-1}$;
otherwise, there exists $y\in \intr \bigcap_{i=1}^{m} H_i$ such that $f_i(x^* + \lambda y) \geq f_i(x^*)$
for any $i\in \{1,\ldots, m\}$ and any $\lambda >0$ with $x^*+ \lambda y \in C$.

Hence, by Lemma~\ref{lemma:hemisphere}, $\bigcap_{i=1}^{m} H_i$
has a nonempty intersection including a great $(d-m)$-sphere $G$ on $S^{d-1}$.
Let $\varphi$ be the corresponding $(d-m+1)$-flat in $\Real^d$ through $x^*$
defined as
 \[\varphi := \{ x^* + \lambda y \in\Real^d \mid y\in G \text{ and } \lambda \in\Real\}.\]

Consider the restriction $f_i|_{\varphi\cap C}$ of $f_i$ on $\varphi \cap C$.
Since $f_i$ is convex and $\varphi$ is an affine subspace (thus, convex),
$f_i|_{\varphi\cap C}$ is also convex and their pointwise minimum $g|_{\varphi\cap C}$
attains a local maximum at $x^* \in \varphi \cap C$.
On the other hand, each $f_i|_{\varphi\cap C}$ attains a local minimum at $x^*$;
since $\varphi \subseteq h_i^\oplus$, we have $f_i(x^*) \leq f_i(x)$
for any point $x\in \varphi \cap C$.
Hence, $g|_{\varphi\cap C}$ also attains a local minimum at $x^*$
since $g(x^*) = f_i(x^*)$ for any $i \in \{1,\ldots,m\}$.
Consequently, $g$ is locally constant at $x^*$ on $\varphi$;
more precisely, there is a sufficiently small neighborhood $U \subset \Real^d$ of $x^*$
with $U \subset C$ such that $g$ is constant on $U \cap \varphi$,
completing the proof.
\end{proof}


\section{Properties of Geodesic-Maximal Pairs} \label{sec:charac}
We call a pair $(s^*,t^*) \in \Poly \times \Poly$ \emph{maximal} if $(s^*,t^*)$
is a local maximum of the geodesic distance function $\dist$.
That is, $(s^*,t^*)$ is maximal if and only if there are two neighborhoods
$U_s, U_t \subset \Real^2$ of $s^*$ and of $t^*$, respectively,
such that for any $s \in U_s\cap \Poly$ and any $t\in U_t \cap \Poly$
we have $\dist(s^*, t^*) \geq \dist(s,t)$.
Clearly, any diametral pair is maximal.

Consider any maximal pair $(s^*, t^*)$ in $\Poly$.
Let $\Pi(s^*,t^*)$ be the set of all shortest paths from $s^*$ to $t^*$.
Then, each path $\pi \in \Pi(s^*, t^*)$ is associated with
a pair of corners $(u,v)$ that are its first and last corners as discussed in Section~\ref{sec:pre}.
Note that such a pair $(u,v)$ of corners always exists for any $\pi \in \Pi(s^*, t^*)$;
even if $\dist(s^*, t^*) = \|s^* - t^*\|$, then both endpoints $s^*$ and $t^*$ must be
corners in $V$ by its maximality.
We now focus on the set of such pairs of the first and last corners, defined to be
 \[ \Vpair(s^*,t^*) := \{ (u,v) \mid \text{$\exists\pi\in\Pi(s^*, t^*)$ s.t. $u, v\in V$ are the first and last corners along $\pi$, resp.} \}.\]
Giving an arbitrary ordering, we set $\Vpair(s^*,t^*) = \{(u_1, v_1), \ldots, (u_m, v_m)\}$,
where $m$ is the cardinality of $\Vpair(s^*, t^*)$.
Also, we let
 \[ V_{s^*} := \{ u_1, \ldots, u_m\}, \quad V_{t^*} := \{v_1,\ldots, v_m\}.\]
Some immediate bounds are
$|\Pi(s^*,t^*)| \geq m$, $|V_{s^*}| \leq m$, and $|V_{t^*}| \leq m$.
Observe that it is not true that we always have the equality $|\Pi(s^*,t^*)| = m$;
in some cases, there can be multiple shortest paths between a pair of corners.
In the following, we show the tight bound on the cardinality $m$ of the set $\Vpair(s^*, t^*)$, provided that $(s^*, t^*)$ is maximal.

Let $E$ be the set of all sides of $\Poly$ without their endpoints and $\Bd$ be their union.
Note that $\Bd = \bd \Poly \setminus V$ is the boundary of $\Poly$ except the corners $V$.
The goal of this section is to prove the following theorem, which is the main combinatorial
result of this paper.
\begin{theorem} \label{theorem:charact}
 Suppose that $(s^*,t^*)$ is a maximal pair in $\Poly$, and that
 $\Vpair(s^*,t^*)$, $V_{s^*}$, and $V_{t^*}$ are defined as above.
 Then, we have the following implications.
 \begin{alignat*}{7}
  &\text{\bf (V-V)} \qquad & &s^*\in V \text{,~ } & & t^*\in V & \quad \text{implies}\quad & |\Vpair(s^*,t^*)|\geq 1, |V_{s^*}|\geq 1, |V_{t^*}|\geq 1;\\
  &\text{\bf (V-B)} \qquad & &s^*\in V \text{,~ } & & t^*\in \Bd & \quad \text{implies}\quad &|\Vpair(s^*,t^*)|\geq 2, |V_{s^*}|\geq 1, |V_{t^*}|\geq 2;\\
  &\text{\bf (V-I)} \qquad & &s^*\in V \text{,~ } & & t^*\in \intr \Poly & \quad \text{implies}\quad &|\Vpair(s^*,t^*)|\geq 3, |V_{s^*}|\geq 1, |V_{t^*}|\geq 3;\\
  &\text{\bf (B-B)} \qquad & &s^*\in \Bd \text{,~ } & & t^*\in \Bd & \quad \text{implies}\quad &|\Vpair(s^*,t^*)|\geq 3, |V_{s^*}|\geq 2, |V_{t^*}|\geq 2;\\
  &\text{\bf (B-I)} \qquad & &s^*\in \Bd \text{,~ } & & t^*\in \intr \Poly & \quad \text{implies}\quad &|\Vpair(s^*,t^*)|\geq 4, |V_{s^*}|\geq 2, |V_{t^*}|\geq 3;\\
  &\text{\bf (I-I)} \qquad & &s^*\in \intr \Poly \text{,~ } & & t^*\in \intr \Poly & \quad \text{implies}\quad &|\Vpair(s^*,t^*)|\geq 5, |V_{s^*}|\geq 3, |V_{t^*}|\geq 3.
 \end{alignat*}
 Moreover, each of the above bounds is tight.
\end{theorem}

Together with the bound $|\Pi(s^*, t^*)| \geq |\Vpair(s^*,t^*)|$,
Theorem~\ref{theorem:charact} immediately implies tight lower bounds
on the number of shortest paths between any maximal pair.
\begin{corollary}
 For any $p\in \Poly$, let $\delta(p) := 0$ if $p\in V$; $\delta(p):=1$ if $p\in \Bd$;
 $\delta(p) := 2$ if $p\in \intr \Poly$.
 If $(s^*, t^*)$ is a maximal pair in $\Poly$, then we have
 \[ |\Pi(s^*, t^*)| \geq \delta(s^*) + \delta(t^*) + 1.\]
 Moreover, the above bound is tight. \hfill\copy\ProofSym
\end{corollary}

To see the tightness of the bounds, we present examples with remarks in \figurename~\ref{fig:examples1} and Appendix A\@.
In particular, one can easily see the tightness of the bounds
on $|V_{s^*}|$ and $|V_{t^*}|$ from shortest path maps $\SPM(s^*)$ and $\SPM(t^*)$,
when $V \cup \{s^*,t^*\}$ is in general position.

We first give an overview of the proof.
The general reasoning is roughly the same for all the different scenarios,
and we thus focus on the case in which $(s^*,t^*)$ is a maximal pair and both $s^*$ and $t^*$ are interior points (Case \textbf{(I-I)}).
Regard the geodesic distance function $\dist$ as a four-variate function
in a small convex neighborhood of $(s^*,t^*)$.
As mentioned in Section \ref{sec:pre}, the geodesic distance is
the pointwise minimum of a finite number of path-length functions.
Since the pair $(s^*,t^*)$ is maximal, we will apply Theorem \ref{theorem:convex_max} and obtain that
the geodesic distance is constant in a flat of dimension $d+1-m=5-m$,
where $m = |\Vpair(s^*, t^*)|$.
On the other hand, we will also show that
the geodesic distance function can only remain constant
in a zero-dimensional flat (i.e., at a point), hence $m\geq 5$.
In the other cases (boundary-interior, boundary-boundary, etc.) the boundary of $\Poly$ introduces additional constraints that reduce the degrees of freedom of the geodesic distance function.
Hence, fewer paths are enough to {\em pin} the solution.

The main technical difficulty of the proof is the fact that
the path-length functions $\plf_{u,v}$ are not globally defined.
Thus, we must properly extend them in a way that all conditions of Theorem~\ref{theorem:convex_max} are satisfied.

\subsection{Proof of Theorem~\ref{theorem:charact}}
We start with several basic observations. The proof of Theorem~\ref{theorem:charact} will be done separately for each case.


The following lemma proves the bounds on $|V_{s^*}|$ and $|V_{t^*}|$ of Theorem~\ref{theorem:charact}.
\begin{lemma}\label{lemma:bounds_anchor}
 Let $(s^*,t^*)$ be a maximal pair.
 \begin{enumerate}
 \item If $t^*\in \Bd$, then $|V_{t^*}|\geq 2$.
 Moreover, if $t^*\in e \in E$, then there exists $v \in V_{t^*}$ such that $v$ is off the line supporting $e$.
 \item If $t^*\in \intr\Poly$, then $|V_{t^*}|\geq 3$ and $t^*$ lies in the interior of the convex hull of $V_{t^*}$.
 \end{enumerate}
\end{lemma}
\begin{proof}
Since $(s^*,t^*)$ is a maximal pair, the function $d_{s^*}(t):=\dist(s^*,t)$ is maximized at $t=t^*$ on a sufficiently small subset $U \subset \Poly$ with $t^*\in U$.
As discussed in Section~\ref{sec:pre}, if $t^*\notin V$, then $t^*$ must be either
a vertex of $\SPM(s^*)$ or an intersection point between an edge of $\SPM(s^*)$ and $\bd \Poly$.
If $t^*\in \intr\Poly$, then $t^*$ should fall into the former case
and hence we have at least three corners $v_1,v_2,v_3 \in V$
determining the vertex $t^*$ of $\SPM(s^*)$.
If $t^* \in \Bd$, then $t^*$ may also occur at the latter case.
In that case, $t^*$ lies on an edge of $\SPM(s^*)$ and thus we have
at least two corners $v_1,v_2 \in V$ determining an edge of $\SPM(s^*)$.

The other claims of the lemma can be shown as follows.
If $t^* \in\intr\Poly$ but $t^*$ lies out of the interior of the convex hull of $V_{t^*}$,
then we can find another point $t \in \Poly$ arbitrarily close to $t^*$ such that $\|t-v_i\| > \|t^*-v_i\|$ for every $v_i\in V_{t^*}$.
This implies that $\dist(s^*,t) > \dist(s^*, t^*)$, contradicting the maximality of $(s^*,t^*)$.
If $t^* \in e \in E$ but every $v_i \in V_{t^*}$ lies on the supporting line $\ell$ of $e$,
then we obtain a strictly larger distance than $\dist(s^*,t^*)$,
as moving $t^*$ in a perpendicular direction to $\ell$.
(Notice that a similar argument can be also found in~\cite[Lemma 2.2]{os-cgd3p-89}.)
\end{proof}

Lemma~\ref{lemma:bounds_anchor} immediately implies the lower bound on $|\Vpair(s^*,t^*)|$
when $s^* \in V$ or $t^* \in V$
since $|\Vpair(s^*,t^*)| \geq \max\{ |V_{s^*}|,|V_{t^*}|\}$.
This completes Cases \textbf{(V-*)}.
Note that the bounds for Case \textbf{(V-V)} are trivial.

{From} now on, we assume that neither $s^*$ nor $t^*$ is a corner of $\Poly$.
This assumption, together with Lemma~\ref{lemma:bounds_anchor}, implies
multiple shortest paths between $s^*$ and $t^*$, and thus $\dist(s^*,t^*) > \|s^*-t^*\|$.
Hence, as discussed in Section~\ref{sec:pre}, any maximal pair falling into one of
Cases \textbf{(B-B)}, \textbf{(B-I)}, and \textbf{(I-I)} appears
as a local maximum of the lower envelope of some path-length functions.

\paragraph{Case \textbf{(I-I)}: When both $s^*$ and $t^*$ lie in $\intr\Poly$.}
We will apply Theorem~\ref{theorem:convex_max} to prove Theorem~\ref{theorem:charact}
for Case \textbf{(I-I)}.
Recall the definition of $\Vpair(s^*, t^*) = \{(u_1, v_1), \ldots, (u_m, v_m)\}$
and $m = |\Vpair(s^*,t^*)|$.
For each $(u_i, v_i) \in \Vpair(s^*,t^*)$,
we have the corresponding shortest path $\pi_i$ between $s^*$ and $t^*$
and $\plf_{u_i,v_i}(s^*,t^*) = \dist(s^*,t^*)$.
Thus, we have at least $m$ functions $f$ among $\{\plf_{u,v} \mid u,v\in V\}$
such that $f(s^*, t^*) = \dist(s^*, t^*)$.
If the number of such path-length functions are exactly $m$,
we can apply Theorem~\ref{theorem:convex_max} directly.

Unfortunately, this is not always the case.
A single shortest path $\pi_i \in \Pi(s^*,t^*)$ may give additional pairs $(u,v)$
of corners with $u, v\in \pi_i$ such that
$(u,v) \neq (u_i, v_i)$ and $\plf_{u,v}(s^*,t^*) = \dist(s^*,t^*)$.
This situation can occur even when the corners of $\Poly$ are in general position.
Observe that this happens only when $u, u_i, s^*$ or $v, v_i, t^*$ are collinear.
In order to resolve this problem, we define the \emph{merged} path-length functions that satisfy
all the requirements of Theorem~\ref{theorem:convex_max} even in degenerate cases.

\begin{figure}[t]
\begin{center}
  \includegraphics[width=.90\textwidth]{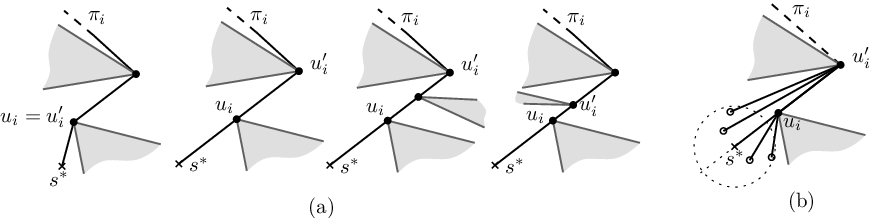}
\end{center}
\caption{\small (a) How to determine $u'_i$. (left to right) $u_i = u'_i$;  $s^*$, $u_i$, and the second corner are collinear; $s^*$ and the first three corners are collinear (b) For points in a small disk $B$ centered at $s^*$ with $B \subset \Vis(u'_i)\cup \Vis(u_i)$, the function $\start_i$ measures the length of the shortest path from $u'_i$ to each.}
\label{fig:degenerate}
\end{figure}

Recall that the combinatorial structure of each shortest path $\pi_i$
can be represented by a sequence $(u_i=u_{i,1},\ldots,u_{i,k}=v_i)$ of corners in $V$.
We define $u'_i$ to be one of the $u_{i,j}$ as follows.
If $s^*$ does not lie on the line $\ell$ through $u_i$ and $u_{i,2}$, then $u'_i := u_i$;
otherwise, if $s^*\in \ell$, then $u'_i := u_{i,j}$, where $j$ is the largest index
such that for any open neighborhood $U \subset \Plane$ of $s^*$
there exists a point $s\in (U \cap \Vis(u_{i,j})) \setminus \ell$.
Note that such $u'_i$ always exists, and if no three of $V$ are collinear,
then we always have either $u'_i=u_i$ or $u'_i=u_{i,2}$.
\figurename~\ref{fig:degenerate}(a) illustrates how to determine $u'_i$.
Also, we define $v'_i$ in an analogous way.
Let $\start_i\colon \Vis(u'_i) \cup \Vis(u_i) \to \Real$ and $\final_i\colon \Vis(v'_i) \cup \Vis(v_i) \to \Real$ be two functions defined as
\begin{eqnarray*}
\start_i(s)&:=& \begin{cases}
             \|s-u'_i\| & \text{if } s\in \Vis(u'_i),\\
             \|s-u_i\| + \|u_i - u'_i\| & \text{if } s\in \Vis(u_i)\setminus \Vis(u'_i);\\
             \end{cases}\\
\final_i(t)&:=& \begin{cases}
             \|t-v'_i\| & \text{if } t\in \Vis(v'_i),\\
             \|t-v_i\| + \|v_i - v'_i\| & \text{if } t\in \Vis(v_i)\setminus \Vis(v'_i).\\
             \end{cases}
\end{eqnarray*}
This allows us to define the merged path-length function $f_i\colon D_i \to \Real$ as
\[  f_i(s,t):=\start_i(s)+\dist(u'_i,v'_i)+\final_i(t), \]
where $D_i:=(\Vis(u'_i) \cup \Vis(u_i)) \times (\Vis(v'_i) \cup \Vis(v_i)) \subseteq \Poly\times\Poly$; see \figurename~\ref{fig:degenerate}(b).
We consider $\Poly \times\Poly$ as a subset of $\Real^4$ and
each pair $(s,t) \in \Poly\times\Poly$ as a point in $\Real^4$.
Also, we denote by $(s_x,s_y)$ the coordinates of a point $s\in\Poly$ and
we write $s=(s_x,s_y)$ or $(s,t) = (s_x, s_y, t_x, t_y)$ by an abuse of notation.
Observe that
\[
f_i(s,t) = \min\{ \plf_{u_i, v_i}(s,t), \plf_{u'_i, v_i}(s,t),
\plf_{u_i, v'_i}(s,t), \plf_{u'_i, v'_i}(s,t) \}
\]
for any $(s,t) \in D_i$ if we define $\plf_{u,v}(s,t)=\infty$ when $s\not\in \Vis(u)$ or $t\not\in \Vis(v)$.

\begin{lemma} \label{lemma:condens}
The following properties hold for the functions $f_i$.
\begin{enumerate} [(i)] \denseitems
 \item $f_i(s^*,t^*) = \dist(s^*,t^*)$ for any $i\in\{1,\ldots,m\}$.
 \item There exists a convex neighborhood $C \subset \Real^4$ of $(s^*,t^*)$
  with $C \subseteq \bigcap_{i=1}^{m} D_i$ such that $\dist(s,t)= \min_{i\in\{1,\ldots,m\}} f_i(s,t)$
  for any $(s,t)\in C$.
 \item Each of the functions $f_i$ for $i \in \{1,\ldots, m\}$ is convex on $C$. 
 \item For any $i\in\{1,\ldots,m\}$, 
  there exists a unique line $\ell_i \subset \Real^4$ through $(s^*, t^*) \in \Real^4$
  such that $f_i$ is constant on $\ell_i \cap C$. 
  Moreover, there exists at most one index $j\neq i$ such that $\ell_i=\ell_j$.
 \item 
  For any $i,j\in\{1,\ldots,m\}$, any $(s,t) \in C$,
  and any neighborhood $U\subseteq C$ of $(s,t)$,
  there exists $(s',t')\in U$ such that
 $f_i(s,t)<f_i(s',t')$ and $f_j(s,t)<f_j(s',t')$.
\end{enumerate}
\end{lemma}
\begin{proof}
\noindent {\bf (i)}
This immediately follows from the fact that $f_i(s^*,t^*)=\plf_{u_i,v_i}(s^*,t^*)$.

\medskip
\noindent {\bf (ii)}
In this proof, we extend $\plf_{u,v}$ to any $(s,t)\in \Poly \times \Poly$
where $\plf_{u,v}(s,t) = \infty$ if $s\notin \Vis(u)$ or $t\notin \Vis(v)$.
By the definition of $f_i$, there exists a small neighborhood $U_i \subset D_i$ of $(s^*,t^*)$ such that $f_i(s,t) = \min\{ \plf_{u_i, v_i}(s,t), \plf_{u'_i, v_i}(s,t), \plf_{u_i, v'_i}(s,t), \plf_{u'_i, v'_i}(s,t) \} = \min_{u,v \in \pi_i \cap V} \plf_{u,v}(s,t)$ for all $(s,t)\in U_i$. We claim that there exists an open convex neighborhood $C \subset \bigcap_i U_i$ such that for any $(s,t)\in C$
\[ \dist(s,t)  = \min_{1\leq i \leq m} f_i(s,t).\]

To prove our claim,
assume to the contrary that for every open convex neighborhood $C\subset \Real^4$
of $(s^*,t^*) \in \Real^4$ there exist a pair $(u,v)$ of corners and $(s,t) \in C$
such that $\dist(s,t) = \plf_{u,v}(s,t) < \min_i f_i(s,t)$.
Note that none of the shortest paths $\pi_i \in \Pi(s^*,t^*)$ between $s^*$ and $t^*$
pass through both of such $u$ and $v$
since, otherwise, we must have $(u, v)\in \Vpair(s^*, t^*)$ and thus $(u,v)=(u_j, v_j)$
for some $1\leq j \leq m$.
This implies that $\dist(s,t) = \plf_{u_j,v_j}(s,t) < \min_i f_i(s,t) = \dist(s,t)$,
a contradiction.


Consider a sequence $C_{1}, C_{2}, \ldots$ of neighborhoods of $(s^*,t^*) \in \Real^4$ that converges to the singleton $\{(s^*,t^*)\}$. Since there are only $n^2$ pairs of corners, there exist a fixed pair $(u_0,v_0)$ of corners and a subsequence $C_{k_1},C_{k_2},\ldots$ converging to the singleton $\{(s^*,t^*)\}$
such that none of the $\pi_i$ pass through both $u_0$ and $v_0$, and
for any integer $j>0$
there exists $(s_j,t_j) \in C_{k_j}$
with
\[\dist(s_j,t_j) = \plf_{u_0,v_0}(s_j,t_j) < \min_{1\leq i\leq m} f_i(s_j,t_j).\]
Since $\lim_{j \to \infty} (s_j,t_j) = (s^*, t^*)$,
it holds that
$\lim_{j\to \infty} \dist(s_j,t_j) = \lim_{j\to\infty} \min_i f_i(s_j,t_j) = \dist(s^*,t^*)$
by Property (i).
By the sandwich theorem, we have
\[ \lim_{j\to\infty} \plf_{u_0,v_0}(s_j, t_j) = \plf_{u_0,v_0}(s^*,t^*) = \dist(s^*,t^*).\]
This implies the existence of the $(m{+}1)$-st shortest path between $s^*$ and $t^*$
since none of the $\pi_i \in \Pi(s^*,t^*)$ contains both $u_0$ and $v_0$,
a contradiction.


\medskip
\noindent {\bf (iii)}
Since the sum of convex functions is a convex function, it suffices to show that $\start_i$ and $\final_i$ are convex. More precisely, for any  $(s_1,t_1), (s_2,t_2) \in C$ and $0 \leq \lambda \leq 1$, we have
\begin{align*}
f_i(\lambda(s_1,t_1)+(1-\lambda)(s_2,t_2))
& =    \start_i(\lambda s_1 + (1-\lambda) s_2) +
        \dist(u'_i,v'_i) +
        \final_i(\lambda t_1 + (1-\lambda) t_2)\\
& \leq \lambda \start_i(s_1) + (1-\lambda) \start_i(s_2) +
       \dist(u'_i,v'_i) +
       \lambda \final_i(t_1) + (1-\lambda) \final_i(t_2)\\
& =    \lambda f_i(s_1,t_1) + (1-\lambda) f_i(s_2,t_2)
\end{align*}
if $\start_i$ and $\final_i$ are convex.

We now show the convexity of $\start_i$
on any convex subset $C \subset \Vis(u'_i)\cup\Vis(u_i)$.
Note that the convexity of $\final_i$ can be shown in the same way.
There are two cases: $u'_i = u_i$ or $u'_i \neq u_i$.
For the former case, $\start_i$ is convex on $C$ since it measures the Euclidean distance
between $u_i$ and a given point in $C$.
For the latter case, let $\ell_0$ be the line through $u_i$, $u'_i$, and also $s^*$.
Then, $C$ may be partitioned by $\ell_0$ into two regions $A_1$ and $A_2$, where $A_1 = C \cap \Vis(u'_i)$ and $A_2 = C\setminus A_1$. Note that $\start_i$ is convex on $A_1$ and on $A_2$. Thus, we are done by checking every point on $\ell_0 \cap C$.

Pick any $s \in \ell_0 \cap C$ and any line $\ell \subset \Real^2$ through $s$.
Let $\theta$ be the angle between $\ell_0$ and $\ell$. If we restrict the domain of $\start_i$ on $\ell \cap C$, then one can check with elementary calculus that both the derivatives of $\|s-u_i\| + \|u_i - u'_i\|$ and of $\|s-u'_i\|$ are equal to $c\cos \theta$ at $s$ for some constant $c$. Hence, $\start_i$ is smooth and convex along $\ell$.
Since we have taken any line $\ell$ through any point on $\ell_0 \cap C$, this suffices to prove the convexity of $\start_i$ on $C$.


\begin{figure}
\begin{center}
  \includegraphics[width=.80\textwidth]{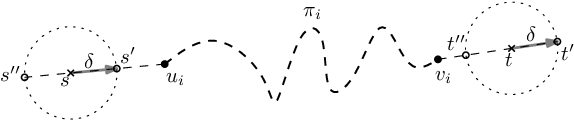}
\end{center}
\caption{\small Illustration to Lemma~\ref{lemma:condens}(v);
for any $(s,t) \in D_i$ and any sufficiently small $\delta$
if we pick $(s',t')$ such that $s'$ is $\delta$ closer to $u_i$ than $s$ and
$t'$ is $\delta$ farther from $v_i$ than $t$,
then we have $f_i(s',t')=f_i(s,t)$. Symmetrically, $f_i(s'', t'') = f_i(s,t)$
with $\|s'' - u_i \| = \|s-u_i\| + \delta$ and $\|t'' - v_i \| = \|t-v_i\| - \delta$.}
\label{fig:uniqueline}
\end{figure}

\medskip
\noindent {\bf (iv)}
Fix any $i\in\{1,\ldots,m\}$.
Any ray $\gamma \subset \Real^4$ with endpoint $(s^*,t^*) \in \Real^4$ can be
determined by three parameters $(\theta_{s^*}, \theta_{t^*}, \lambda)$
with $0\leq \theta_{s^*}, \theta_{t^*} \leq \pi$ and $\lambda \geq 0$ as follows:
Let $\gamma_{s^*}$ and $\gamma_{t^*}$ be the projections of $\gamma$ onto
the $(s_x,s_y)$-plane and the $(t_x,t_y)$-plane, respectively.
Note that $\gamma_{s^*}$ is a ray in the $(s_x,s_y)$-plane with endpoint $s^*$
and $\gamma_{t^*}$ is a ray in the $(t_x,t_y)$-plane with endpoint $t^*$.
Let $\theta_{s^*}$ be the smaller angle at $s^*$
made by $\gamma_{s^*}$ and another ray starting from $s^*$ in direction
away from $u_i$.
Define $\theta_{t^*}$ analogously with $\gamma_{t^*}$, $t^*$, and $v_i$.
The derivative of $f_i$ at $(s^*,t^*)$ along $\gamma$ is represented as
$c(\cos \theta_{s^*} + \lambda \cos\theta_{t^*})$
for some constants $\lambda \geq 0$ and $c > 0$ depending only on $i$.
Note that the second derivative of $f_i$ at $(s^*,t^*)$ along $\gamma$ is derived as
$c\left(\frac{\sin^2\theta_{s^*}}{\|s^*-u_i\|} + \lambda \frac{\sin^2\theta_{t^*}}{\|t^*-v_i\|}\right)$.

Suppose that $f_i$ is constant along $\gamma$ locally around $(s^*,t^*)$.
Then, its first and second derivatives along $\gamma$ should be zero
in a small neighborhood $U \subset\Real^4$ of $(s^*,t^*)$ with $U\subset D_i$.
First, we observe that $\lambda$ should be positive;
if $\lambda = 0$, then $t=t^*$ is fixed while $s$ moves from $s^*$ along $\gamma_{s^*}$,
and hence $f_i$ does not stay constant.
Since every term of the second derivative is nonnegative and $\lambda>0$,
we only obtain two solutions $(\theta_{s^*},\theta_{t^*})=(0,\pi)$ or $(\pi,0)$.
Consequently, we have two such rays $\gamma = (0, \pi, 1)$ or $(\pi, 0, 1)$
that $f_i$ remains constant along $\gamma$.
These two rays form a unique line $\ell_i \subset \Real^4$ through $(s,t)$ such that $f_i$ is constant on $\ell_i \cap U$.
See \figurename~\ref{fig:uniqueline} for more intuitive and geometric description of $\ell_i$.

The projections of $\ell_i$ onto the $(s_x, s_y)$-plane and the $(t_x, t_y)$-plane
appear the lines through $s^*$ and $u_i$ and through $t^*$ and $v_i$, respectively.
Hence, one can easily check that $f_i$ remains constant on $\ell_i \cap D_i$,
which completes the proof of the first part of the claim.

We now show the second part of the claim.
As observed above,
we have that the projection of $\ell_i$ onto the $(s_x, s_y)$-plane is the line through $s^*$ and $u_i$.
Also, the projection of $\ell_i$ onto the $(t_x,t_y)$-plane is the line through $t^*$ and $v_i$.
Hence, $\ell_i=\ell_j$ implies that
$u_i$, $u_j$, $s^*$ are collinear and $v_i$, $v_j$, $t^*$ are collinear.
First, since the pairs $(u_i, v_i)$ are all distinct,
we have $u_i \neq u_j$ or $v_i \neq v_j$.
If $u_i = u_j$ and $v_i\neq v_j$, then one can easily check that $\ell_i \neq \ell_j$
from geometric interpretation of $\ell_i$ as shown in~\figurename~\ref{fig:uniqueline}.
We hence have $u_i \neq u_j$ and $v_i\neq v_j$.
Moreover, $s^*$ must lie in between $u_i$ and $u_j$ and
$t^*$ must lie in between $v_i$ and $v_j$ by definition;
if $u_j$ lies in between $u_i$ and $s^*$, then the first corner of $\pi_i$ from $s^*$
becomes $u_j$ since the three are collinear.
Therefore, for each $i\in\{1,\ldots, m\}$, there is at most one index $j\in\{1,\ldots, m\}$
such that $i\neq j$ and $\ell_i = \ell_j$.

\medskip
\noindent {\bf (v)}
Pick any $i, j \in \{1, \ldots, m\}$ and
consider the sublevel sets
$L_i = \{(\tilde{s},\tilde{t})\in \Real^4 \mid f_i(\tilde{s},\tilde{t}) \leq f_i(s,t)\}$
and
$L_j = \{(\tilde{s},\tilde{t})\in \Real^4 \mid f_j(\tilde{s},\tilde{t}) \leq f_j(s,t)\}$.
Since $f_i$ and $f_j$ are convex and non-constant functions,
$L_i$ and $L_j$ are closed convex sets that have $(s,t)$ on their boundaries.
Therefore, there exist hyperplanes $h_i$ and $h_j$ tangent to
$L_i$ and $L_j$, respectively, at $(s,t)$.
Let $h_i^{\oplus}$ be a closed half-space
bounded by $h_i$ that avoids $L_i$ and
$H'_i:=\{(s',t')\in h_i^\oplus \mid \|(s',t')-(s,t)\|=1\}$
be a closed hemisphere on the unit sphere centered at $(s,t)$.
Define $H'_j$ analogously for $h_j$.

Since $H'_i$ and $H'_j$ are closed hemispheres with a common center,
$H'_i \cap H'_j \neq \emptyset$.
By construction, we have $f_i(s,t)\leq f_i(s',t')$ for any $(s',t')\in H'_i$,
and $f_j(s,t)\leq f_j(s',t')$ for any $(s',t')\in H'_j$.
On the other hand, by Property (iv) of the lemma,
the equality holds only when $(s', t')$ lies on line $\ell_i$ or $\ell_j$, respectively.
Therefore, for any $(s', t') \in (H'_i \cap H'_j) \setminus (\ell_i \cup \ell_j)$,
the claimed inequalities $f_i(s,t) < f_i(s',t')$ and $f_j(s,t) < f_j(s',t')$
hold strictly.
The last task is to check that
$(H'_i \cap H'_j) \setminus (\ell_i \cup \ell_j) \neq \emptyset$,
which follows clear by Lemma~\ref{lemma:hemisphere}.
\end{proof}

Back to the proof of Theorem~\ref{theorem:charact},
we take a convex neighborhood $C$ of $(s^*, t^*)$ satisfying Property (ii)
of Lemma~\ref{lemma:condens} and apply Theorem~\ref{theorem:convex_max}.
Note that Properties (i)--(iii) of Lemma~\ref{lemma:condens}
ensure that the preconditions of Theorem~\ref{theorem:convex_max} are satisfied.

\begin{figure}
\begin{center}
  \includegraphics[width=.60\textwidth]{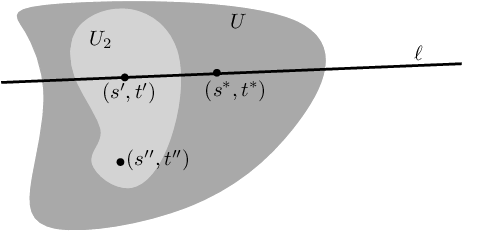}
\end{center}
\caption{\small Proof of Theorem~\ref{theorem:charact} for the \textbf{(I-I)} case; If $\dist(s,t)$ is constant on $\ell$, we can find pairs of points $(s'',t'')$ arbitrarily close to $(s^*,t^*)$ whose geodesic distance is larger than $\dist(s^*,t^*)$.}
\label{fig:local}
\end{figure}

Suppose that $m < 5$.
Then, by Theorem~\ref{theorem:convex_max},
there exists at least one line $\ell \in \Real^4$ through $(s^*,t^*)$ such that
$\dist$ is constant on $\ell \cap C$.
Since $(s^*,t^*)$ is a local maximum, there exists a small neighborhood $U \subset C$
of $(s^*,t^*)$ such that $\dist(s,t)\leq \dist(s^*,t^*)$ for all $(s,t)\in U$.
By Property (iv) of Lemma~\ref{lemma:condens}, at most two functions $f_i$ are constant
on $\ell\cap U$.
Without loss of generality, we can assume that functions $f_3,\ldots, f_m$ are not constant.
Since the geodesic distance function $\dist$ is constant on $\ell\cap U$ and
$\dist(s,t)= \min_{i\in\{1,\ldots,m\}} f_i(s,t)$, any of $f_3, \ldots, f_m$
must strictly increase in both directions along $\ell$.
That is, for any $(s',t')\in \ell\cap U$ with $(s',t')\neq(s^*,t^*)$
and for all $i\geq 3$, we have $\min\{f_1(s',t'),f_2(s',t')\}<f_i(s',t')$.
Thus, there exists a small neighborhood $U'\subseteq U$ of $(s',t')$
such that $\dist(s,t)= \min\{f_1(s,t),f_2(s,t)\}$ for all $(s,t)\in U'$.
However, by Property (v) of Lemma~\ref{lemma:condens},
there exists a pair $(s'',t'')\in U'$ such that
$f_1(s',t')<f_1(s'',t'')$ and $f_2(s',t')<f_2(s'',t'')$,
contradicting the maximality of $(s^*,t^*)$.
See Figure~\ref{fig:local}.
Hence, we achieve a bound $m=|\Vpair(s^*,t^*)| \geq 5$,
as claimed in Case \textbf{(I-I)} of Theorem~\ref{theorem:charact}.


\paragraph{Case \textbf{(B-B)}: When both $s^*$ and $t^*$ lie on $\Bd$.}
In this case, we assume that $s^*\in e_s \in E$ and $t^*\in e_t \in E$.
The outline of proof is analogous to the above discussion for Case \textbf{(I-I)};
the only difference is that the search space has a lower dimension.

Let $p$ be an endpoint of $e_s$ and $l_s$ be the length of $e_s$.
We denote by $s(\zeta_s)$ the unique point on $e_s$ such that $\|s(\zeta_s) - p\| = \zeta_s$ for any $0 < \zeta_s < l_s$.
Thus, $s \colon (0, l_s) \to e_s$ establishes a bijection
between the open interval $(0,l_s) \subset \Real$ and the segment $e_s \subset \Real^2$
except its endpoints.
We also define $t(\zeta_t)$, analogously.
Then, we let $\bar{f_i} \colon D_i \to \Real$ be a function
defined as the composition of $f_i$ and the two bijections:
\[ \bar{f_i}(\zeta_s, \zeta_t) := \start_i(s(\zeta_s)) + \dist(u'_i,v'_i) + \final_i(t(\zeta_t)),\]
where the domain of $\bar{f_i}$ is $D_i := s^{-1}((\Vis(u'_i) \cup \Vis(u_i)) \cap e_s) \times t^{-1}((\Vis(v'_i) \cup \Vis(v_i)) \cap e_t)$. We consider $D_i$ as a subset of $\Real^2$ and each pair $(\zeta_s, \zeta_t) \in D_i$ as a point in $\Real^2$. Let $\zeta^*_s$ and $\zeta^*_t$ be real numbers such that $s^* = s(\zeta^*_s)$ and $t^* = t(\zeta^*_t)$. We obtain the analogue of Lemma \ref{lemma:condens}.

\begin{lemma} \label{lemma:condensBB}
The following properties hold for the functions $\bar{f_i}$.
\begin{enumerate} [(i)] \denseitems
 \item $\bar{f_i}(\zeta_s^*,\zeta_t^*) = \dist(s(\zeta_s^*), t(\zeta_t^*))$ for any $i\in\{1,\ldots,m\}$.
 \item There exists a convex neighborhood $C \subset \Real^2$ of $(\zeta^*_s,\zeta^*_t)$
with $C \subseteq \bigcap_{i=1}^{m} D_i$ such that $\dist(s(\zeta_s), t(\zeta_t))= \min_{i\in\{1,\ldots,m\}} \bar{f_i}(\zeta_s, \zeta_t)$ for any $(\zeta_s, \zeta_t)\in C$.
\item Each of the functions $\bar{f_i}$ for $i \in \{1,\ldots, m\}$ is convex on $C$.
\item If there exists a  line $\ell_i \subset \Real^2$ such that $\bar{f_i}$ is constant on $\ell_i \cap C$, then $u_i$ lies on the line supporting $e_s$  and $v_i$ lies on the line supporting $e_t$.
 \item For any $i\in\{1,\ldots,m\}$, any $(\zeta_s,\zeta_t) \in C$, and
  any neighborhood $U\subseteq C$ of $(\zeta_s,\zeta_t)$,
  there exists $(\zeta_s',\zeta_t')\in U$ such that
  $\bar{f_i}(\zeta_s,\zeta_t)<\bar{f_i}(\zeta_s',\zeta_t')$.
\end{enumerate}
\end{lemma}
Note that the above claims are almost identical to those of Lemma \ref{lemma:condens}. The results have been adapted taking into account that $\bar{f_i}$ is the composition of $f_i$ and both $\zeta_s$ and $\zeta_t$. Proofs follow verbatim, thus we omit them. Property (v) is the only exception: since the degrees of freedom have decreased, we cannot certify the existence of points arbitrarily close that increase two functions $\bar{f_i}$. Instead, we will use the second property of Lemma \ref{lemma:bounds_anchor} to lead to a contradiction.

Recall that by the first claim of Lemma~\ref{lemma:bounds_anchor} we have $m\geq 2$.
Thus, we are done by showing that the case $m=2$ is not possible.
Suppose that $m=2$.
Then, by Theorem~\ref{theorem:convex_max},
there exists a line $\ell \subset \Real^2$ through $(\zeta^*_s, \zeta^*_t) \in \Real^2$
such that $\dist$ is constant on $\ell \cap C$.
By the second claim of Lemma~\ref{lemma:bounds_anchor},
there exists a vertex $v\in V_{s^*}$ off the line supporting $e_s$.
Without loss of generality, we assume that $v = v_2$. 
By Property (iv) of Lemma~\ref{lemma:condensBB},
function $\bar{f_2}$ cannot remain constant in any line.

Now, we proceed as in Case \textbf{(I-I)}.
Consider any small neighborhood $U\subseteq C$ of $(\zeta_s^*, \zeta_t^*)$.
Any point $(\zeta_s', \zeta_t')\in\ell\cap C$ with $(\zeta_s', \zeta_t')\neq (\zeta^*_s, \zeta^*_t)$ satisfies the strict inequality $\dist(s(\zeta_s'), t(\zeta_t'))=\bar{f_1}(\zeta_s', \zeta_t') < \bar{f_2}(\zeta_s', \zeta_t')$,
since $\bar{f_2}$ cannot remain constant and $\dist$ is a local maximum.
Thus, there exists a sufficiently small neighborhood $U' \subseteq U$
of $(\zeta_s', \zeta_t')$ such that
$\dist(s(\zeta_s), t(\zeta_t))=\bar{f_1}(\zeta_s, \zeta_t)$
for all $(\zeta_s, \zeta_t)\in U'$.

Now, we apply Property (v) of Lemma~\ref{lemma:condensBB} to obtain
a point $(\zeta_s'', \zeta_t'')$ arbitrarily close to $(\zeta_s^*, \zeta_t^*)$
with strict inequality $\dist(s(\zeta_s''), t(\zeta_t'')) > \dist(s(\zeta_s^*), t(\zeta_t^*)) = \dist(s^*, t^*)$,
contradicting the maximality of $(s^*, t^*)$.
We hence conclude that $m=|\Vpair(s^*,t^*)| \geq 3$
for Case \textbf{(B-B)} when both $s^*$ and $t^*$ lie on $\Bd$.

\paragraph{Case \textbf{(B-I)}: When $s^* \in \Bd$ and $t^*\in \intr\Poly$.}
This case is a mixture of the two previous cases.
Without loss of generality, we can also assume that $s^* \in e_s \in E$ and $t^* \in \intr\Poly$.
We define $s(\zeta_s)$ as in Case \textbf{(B-B)} with $s(\zeta^*_s) = s^*$. We now define function $\hat{f_i} \colon D_i \to \Real$ as $\hat{ f_i}(\zeta_s, t_x, t_y) := \start_i(s(\zeta_s)) + \dist(u'_i, v'_i) + \final_i(t_x,t_y)$, where $D_i := s^{-1}((\Vis(u'_i) \cup \Vis(u_i)) \cap e_s) \times (\Vis(v'_i) \cup \Vis(v_i))$ is a subset of $\Real^3$.

We obtain another analogy of Lemmas~\ref{lemma:condens} and~\ref{lemma:condensBB}.
\begin{lemma} \label{lemma:condensBI}
The following properties hold for the functions $\hat{f_i}$.
\begin{enumerate} [(i)] \denseitems
 \item $\hat{f_i}(\zeta_s^*,t^*) = \dist(s(\zeta_s^*), t)$ for any $i\in\{1,\ldots,m\}$.
 \item There exists a convex neighborhood $C \subset \Real^3$ of $(\zeta^*_s,t^*)$
with $C \subseteq \bigcap_{i=1}^{m} D_i$ such that $\dist(s(\zeta_s), t)= \min_{i\in\{1,\ldots,m\}} \hat{f_i}(\zeta_s, t)$ for any $(\zeta_s, t)\in C$.
\item Each of the functions $\hat{f_i}$ for $i\in \{1,\ldots,m\}$ is convex on $C$.
\item For any $i\in\{1,\ldots,m\}$, 
 there exists a unique line $\ell_i \subset \Real^3$ through $(\zeta_s^*, t^*) \in \Real^3$
 such that $\hat{f_i}$ is constant on $\ell_i \cap C$.
 Moreover, there is at most one index $j\neq i$ such that $\ell_i=\ell_j$.
 \item For any $i,j\in\{1,\ldots,m\}$, any $(\zeta_s,t) \in C$, and
 any neighborhood $U\subseteq C$ of $(\zeta_s,t)$,
 there exists $(\zeta_{s}',t')\in U$
 such that $\hat{f_i}(\zeta_s,t)<\hat{f_i}(\zeta_{s}',t')$
 and $\hat{f_j}(\zeta_s,t)<\hat{f_j}(\zeta_{s}',t')$.
\end{enumerate}
\end{lemma}
We proceed as in Case~\textbf{(B-B)}.
Suppose $m\leq 3$ and apply Theorem~\ref{theorem:convex_max}.
Then, we obtain a line $\ell$ such that the geodesic distance (composed with $\zeta_s$)
is constant on $\ell \cap C$.
However, since at most two functions $f_i$ can remain constant on $\ell$
by Property (iv) of Lemma~\ref{lemma:condensBI},
there must exist a point arbitrarily close to $(\zeta^*_s,t^*)$
with strictly larger function value.
Details are almost identical to the previous cases,
and we get the claimed bound $m = |\Vpair(s^*, t^*)| \geq 4$
for Case~\textbf{(B-I)}.

The claimed bounds on $|V_{s^*}|$ and $|V_{t^*}|$ are shown by Lemma~\ref{lemma:bounds_anchor}, which completes the proof of Theorem~\ref{theorem:charact}.
\hfill\copy\ProofSym

\section{Computing the Geodesic Diameter} \label{sec:algorithm}
Since a diametral pair is in fact maximal, it falls into one of the cases shown in Theorem~\ref{theorem:charact}. In order to find a diametral pair we examine all possible scenarios accordingly.

Cases \textbf{(V-*)}, where at least one point is a corner in $V$, can be handled
in $O(n^2 \log n)$ time by computing $\SPM(v)$ for every $v\in V$ and
traversing it to find the farthest point from $v$, as discussed in Section~\ref{sec:pre}.
We thus focus on Cases \textbf{(B-B)}, \textbf{(B-I)}, and \textbf{(I-I)},
where a diametral pair consists of two non-corner points.

From the computational point of view, the most difficult case corresponds to Case \textbf{(I-I)} of Theorem~\ref{theorem:charact}.
In particular, if $|V_{s^*}|=|V_{t^*}|=5$, ten corners of $V$ are involved
and thus any exhaustive method would check $O(n^{10})$ possibilities
to find maximal pairs of this case.
Observe that such a case can happen even under a general position assumption
as shown in Appendix~A.3.
By Theorem~\ref{theorem:charact}, 
in Case \textbf{(I-I)}, it is guaranteed that there are at least five distinct pairs
$(u_1, v_1), \ldots, (u_5,v_5)$ of corners in $V$
such that $\plf_{u_i,v_i}(s^*,t^*) = \dist(s^*,t^*)$ for any $i\in \{1,\ldots,5\}$ and
the system of equations $\plf_{u_1,v_1}(s,t)= \cdots = \plf_{u_5,v_5}(s,t)$
determines a $0$-dimensional zero set,
corresponding to a constant number of candidate pairs in $\intr\Poly\times\intr\Poly$.
On the other hand, each path-length function $\plf_{u,v}$
is an algebraic function of degree at most $4$.
Thus, given five distinct pairs $(u_i,v_i)$ of corners,
we can compute all candidate pairs $(s,t)$ in $O(1)$ time by solving the system.\footnote{%
Here, we assume that fundamental operations on a constant number
of polynomials of constant degree with a constant number of variables
can be performed in constant time.}
For each candidate pair 
we compute the geodesic distance between the pair to check its validity.
Since the geodesic distance between any two points $s,t \in \Poly$ can be computed
in $O(n\log n)$ time~\cite{hs-oaespp-99},
we obtain a brute-force $O(n^{11}\log n)$-time algorithm,
checking $O(n^{10})$ candidate pairs obtained from all possible combinations of $10$ corners in $V$.

As a different approach, one can exploit the \emph{$\SPM$-equivalence decomposition} of $\Poly$,
which subdivides $\Poly$ into regions such that the shortest path map
of any two points in a common region are {\em topologically equivalent}~\cite{cm-tpespqp-99}.
It is not difficult to see that
if $(s,t)$ is a pair of points that equalizes any five path-length functions,
then both $s$ and $t$ appear as vertices of the decomposition.
However, the current best upper bound on the complexity of the $\SPM$-equivalence decomposition
is $O(n^{10})$~\cite{cm-tpespqp-99}, and thus this approach hardly leads to a remarkable improvement.

Instead, we do the following for Case \textbf{(I-I)} with $|V_{s^*}| = 5$. 
We choose any five corners $u_1,\ldots,u_5 \in V$ (as a candidate for the set $V_{s^*}$) and overlay their shortest path maps $\SPM(u_i)$.
Since each $\SPM(u_i)$ has $O(n)$ complexity, the overlay consists of $O(n^2)$ cells.
Any cell of the overlay is the intersection of five cells associated with $v_1,\ldots,v_5 \in V$
in $\SPM(u_1),\ldots,\SPM(u_5)$, respectively.
Choosing a cell of the overlay, we get five (possibly, not distinct) corners $v_1,\ldots, v_5$
and a constant number of candidate pairs by solving the system
$\plf_{u_1,v_1}(s,t)=$ $\cdots = \plf_{u_5,v_5}(s,t)$.
We iterate this process for all possible tuples of five corners $u_1,\ldots,u_5$,
to obtain a total of $O(n^7)$ candidate pairs,
roughly spending $O(n^7 \log n)$ time.
Note that the other subcases with $|V_{s^*}| \leq 4$ can be handled similarly,
resulting in $O(n^6)$ candidate pairs.

The validity of each candidate pair $(s,t)$ is examined
by checking if the paths from $s$ through $u_i$ and $v_i$ to $t$
are indeed shortest.
For the purpose, we evaluate its geodesic distance $\dist(s,t)$
using a two-point query structure of Chiang and Mitchell~\cite{cm-tpespqp-99}.
For a fixed parameter $0<\delta\leq 1$ and any fixed $\epsilon>0$,
one can construct, in $O(n^{5+10\delta+\epsilon})$ time,
a data structure that supports $O(n^{1-\delta}\log n)$-time
two-point shortest path queries.
The total running time is $O(n^7 \log n) + O(n^{5+10\delta+\epsilon}) + O(n^7)\times O(n^{1-\delta}\log n)$.
We set $\delta = \frac{3}{11}$ to optimize the running time
to $O(n^{7+\frac{8}{11}+\epsilon})$.

Also, we can use an alternative two-point query data structure whose performance
is sensitive to the number $h$ of holes~\cite{cm-tpespqp-99}:
after $O(n^5)$ preprocessing time using $O(n^5)$ storage,
two-point queries can be answered in $O(\log n +h)$ time.\footnote{%
If $h$ is relatively small,
one could use the structure of Guo, Maheshwari and Sack~\cite{gms-spqpd-08}
which answers a two-point query in $O(h \log n)$ time after $O(n^2\log n)$ preprocessing time
using $O(n^2)$ storage, or another structure by Chiang and Mitchell~\cite{cm-tpespqp-99}
that supports a two-point query in $O(h\log n)$ time,
spending $O(n+h^5)$ preprocessing time and storage.}
Using this alternative structure,
the total running time of our algorithm amounts to $O(n^7 (\log n + h))$.
Note that this method gives a better bound than the previous one
when $h= O(n^\frac{8}{11})$.

The other cases can be handled analogously with strictly better time bound.
For Case \textbf{(B-I)}, by Theorem~\ref{theorem:charact},
we have $|\Vpair(s^*, t^*)| \geq 4$ and thus there are at least four distinct pairs
$(u_i, v_i)$ of corners with $\plf_{u_i,v_i}(s^*, t^*) = \dist(s^*, t^*)$.
Here, we handle only the case of $|V_{t^*}| = 3$ or $4$.
For the subcase with $|V_{t^*}|= 4$,
we choose any four corners from $V$ as $v_1, \ldots, v_4$ as a candidate for $V_{t^*}$
and overlay their shortest path maps $\SPM(v_i)$.
The overlay, together with $V$, decomposes $\bd\Poly$ into $O(n)$ intervals.
Each such interval determines $u_1,\ldots, u_4$ as above, and the side $e_s \in E$
on which $s^*$ should lie.
Now, we have a system of four equations on four variables:
three from the corresponding path-length functions $\plf_{u_i,v_i}$ with $1 \leq i \leq 4$
which should be equalized at $(s^*,t^*)$, and the fourth from the supporting line of $e_s$.
Solving the system, we get a constant number of candidate maximal pairs,
again by Theorem~\ref{theorem:charact}. 
In total, we obtain $O(n^5)$ candidate pairs.
The other subcase with $|V_{t^*}| = 3$ can be handled similarly,
resulting in $O(n^4)$ candidate pairs.
As above, we can exploit two different structures for two-point queries.
Consequently, we can handle Case \textbf{(B-I)} in $O(n^{5+\frac{10}{11}+\epsilon})$ or $O(n^5 (\log n + h))$ time.

In Case \textbf{(B-B)} when $s^*, t^* \in \Bd$, 
we have $|V_{s^*}| = 2$ or $3$.
For the subcase with $|V_{s^*}|=3$,
we choose three corners as a candidate of $V_{s^*}$ and
take the overlay of their shortest path maps $\SPM(u_i)$.
It decomposes $\bd\Poly$ into $O(n)$ intervals.
Each such interval determines three corners $v_1, v_2, v_3$ forming $V_{t^*}$ and
a side $e_t \in E$ on which $t^*$ should lie.
Note that we have only three equations so far;
two from the three path-length functions and the third from the line supporting to $e_t$.
Since $s^*$ also should lie on a side $e_s\in E$ with $e_s\neq e_t$,
we need to fix such a side $e_s$ that $\bigcap_{1\leq i\leq 3} \Vis(u_i)$ intersects $e_s$.
In the worst case, the number of such sides $e_s$ is $\Theta(n)$.
Thus, we have $O(n^5)$ candidate pairs for Case \textbf{(B-B)};
again, the other subcase with $|V_{s^*}| = 2$ contributes to a smaller number $O(n^4)$ of candidate pairs.
Testing each candidate pair can be 
done as above,
resulting in $O(n^{5+\frac{10}{11}+\epsilon})$ or $O(n^5(\log n + h))$ total running time.

Alternatively, 
one can exploit a two-point query structure only for boundary points on $\bd\Poly$
for Case \textbf{(B-B)}.
The two-point query structure by Bae and Okamato~\cite{bo-qtbpsppd-09}  builds
an explicit representation of the graph of the lower envelope of the path-length functions
$\plf_{u,v}$ restricted on $\bd \Poly \times \bd \Poly $ in $O(n^5 \log n \log^* n)$ time.\footnote{%
More precisely, in $O(n^4\lambda_{65}(n)\log n)$ time,
where $\lambda_{m}(n)$ stands for the maximum length of a Davenport-Schinzel
sequence of order $m$ on $n$ symbols.}
Since $|\Vpair(s^*,t^*)|\geq 3$ in Case \textbf{(B-B)},
such a pair appears as a vertex on the lower envelope.
Hence, we are done by traversing all the vertices of the lower envelope.

The following table summarizes the discussion so far.
\begin{center} \small
\begin{tabular}{|c||c|c|}
\hline
Case     & Independent of $h$ & Dependent on $h$ \\
\hline
\hline
\textbf{(V-*)} & \multicolumn{2}{|c|}{$O(n^2 \log n)$} \\
\hline
\textbf{(B-B)}      &  $O(n^5\log n \log^* n)$ & $O(n^5(\log n + h))$  \\
\hline
\textbf{(B-I)}      & $O(n^{5+\frac{10}{11}+\epsilon})$ & $O(n^5(\log n + h))$ \\
\hline
\textbf{(I-I)}      & $O(n^{7+\frac{8}{11}+\epsilon})$ & $O(n^7 (\log n + h))$\\
\hline
\end{tabular}
\end{center}

As Case \textbf{(I-I)} is the bottleneck, we conclude the following.
\begin{theorem}\label{theorem:alg_diam}
 Given a polygonal domain having $n$ corners and $h$ holes,
 the geodesic diameter and a diametral pair can be computed
 in $O(n^{7+\frac{8}{11}+\epsilon})$ or $O(n^7  (\log n+ h))$ time
 in the worst case, where $\epsilon$ is any fixed positive number.
 \hfill\copy\ProofSym
\end{theorem}

%

\section{Concluding Remarks} \label{sec:conclusion}

We have presented the first algorithms that compute the geodesic diameter of a given polygonal domain. As mentioned in the introduction, a similar result for convex 3-polytopes was shown in~\cite{os-cgd3p-89}. We note that, although the main result of this paper is similar, the techniques used in the proof are quite different. Indeed, the key requirement for our proof is the fact that shortest paths in our environment are polygonal chains whose vertices are in $V$, a claim that does not hold in higher dimensions (even in 2.5-D surfaces). It would be interesting to find other environments in which similar result holds.

Another interesting question would be finding out how many maximal pairs a polygonal domain can have. The analysis of Section~\ref{sec:algorithm} gives an $O(n^7)$ upper bound. On the other hand, one can easily construct a simple polygon in which the number of maximal pairs is $\Omega(n^2)$. Any improvement on the $O(n^7)$ upper bound would lead to an improvement in the running time of our algorithm.

Though in this paper we have focused on \emph{exact} geodesic diameters only,
an efficient algorithm for finding an \emph{approximate} geodesic diameter would be also interesting.
Notice that any point $s\in \Poly$ and its farthest point $t\in\Poly$
yield a 
$2$-approximate diameter;
that is, 
$\diam(\Poly) \leq 2 \max_{t\in \Poly} \dist(s,t)$ for any $s\in\Poly$.
Also, based on a standard technique using a rectangular grid with a specified parameter $0 < \epsilon < 1$,
one can obtain a $(1+\epsilon)$-approximate diameter
in $O((\frac{n}{\epsilon^2}+\frac{n^2}{\epsilon})\log n)$ time as follows.
Scale $\Poly$ so that $\Poly$ can fit into a unit square, and
partition $\Poly$ with a grid of size $\epsilon^{-1}\times \epsilon^{-1}$.
We define the set $D$ as the point set that has the center of grid squares (that have a nonempty intersection with $\Poly$) and intersection points between boundary edges and grid segments. We now can discretize the diameter problem  by considering only geodesic distances between pairs of points of $D$.
It turns out that 
the distance between any two points $s$ and $t$ in $\Poly$ is within a $(1+\epsilon)$ factor
of the distance between two points of $D$.\footnote{%
The idea of this approximation algorithm is due to Hee-Kap Ahn.}
Breaking the quadratic bound in $n$ for the $(1+\epsilon)$-approximate diameter seems
a challenge at this stage.
We conclude by posing the following problem:
for any or some $0 < \epsilon < 1$, is there any algorithm that
finds a $(1+\epsilon)$-approximate diametral pair
in $O(n^{2-\delta}\cdot \mathrm{poly}(1/\epsilon))$ time for some positive $\delta>0$?

\paragraph{Acknowledgements}
We thank Hee-Kap
Ahn, Jiongxin Jin, Christian Knauer, and
Joseph Mitchell for fruitful
discussion.
We also thank 
Joseph O'Rourke for pointing out the reference \cite{z-ipt-07}.



\pagebreak

\appendix
{\large \bfseries \noindent APPENDIX}
\section{More Examples and Remarks}\label{sec:moreexample}
In this section, we show more constructions of polygonal domains and their diametral pairs
with remarks.
In the figures, we keep the following rules:
the boundary $\bd\Poly$ is depicted by dark gray segments and the interior of holes
by light gray region. A diametral pair is given as $(s^*,t^*)$ and shortest paths
between $s^*$ and $t^*$ are described as black dashed polygonal chains.
\subsection{Examples where at least one point of a diametral pair lies on $\bd\Poly$}

\begin{figure}[h]
\begin{center}
  \includegraphics[width=.95\textwidth]{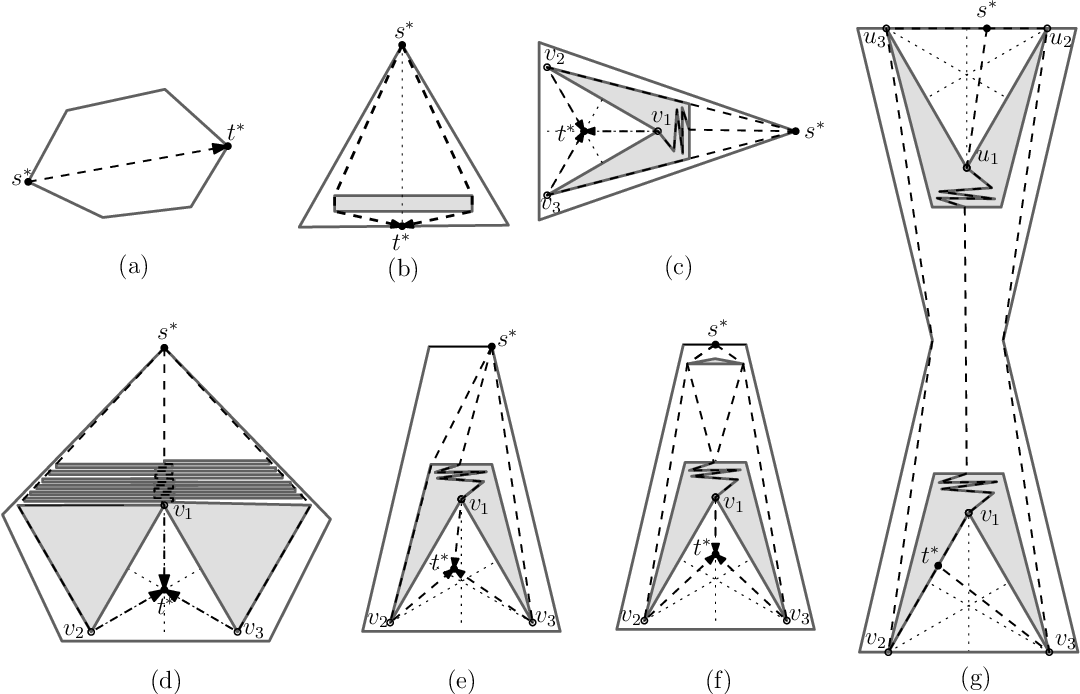}
\end{center}
\caption{ \small
 (a--c) Polygonal domains whose geodesic diameter is determined by a corner $s^*$
 and (d--g) variations of the construction (c).
 (a) When both $s^*$ and $t^*$ are corners;
 (b) When $t^*$ is a point on $\bd\Poly$;
 (c) When $t^* \in \intr\Poly$.
 This polygonal domain consists of two holes, forming a narrow corridor and
 three shortest paths between $s^*$ and $t^*$.
 Here, we have $\dist(s^*, v_1) = \dist(s^*,v_2) = \dist(s^*,v_3)$
 and $t^*$ is indeed the vertex of $\SPM(s^*)$ defined by $v_1, v_2, v_3$;
 (d) Variation of (c) with all convex holes;
 (e) Three shortest paths are not enough to determine a boundary-interior diametral pair;
 (f) If we add one more hole, then the diameter is determined by $s^* \in \Bd$ and $t^*\in\intr\Poly$
  with four shortest paths;
 (g) A polygonal domain made by attaching two copies of (e) and modifying
 it to have $\dist(u_1, v_1) = \dist(u_2,v_2) = \dist(u_3,v_3)$.
 Observe that, in this polygonal domain, the diameter is determined by two boundary points
 with three shortest paths.
 }
\label{fig:examples2}
\end{figure}
%

Note that, as expected, every example in \figurename~\ref{fig:examples2} obeys
Theorem~\ref{theorem:charact}.
An interesting construction is \figurename~\ref{fig:examples2}(g),
where neither of the two centers of $\triangle u_1u_2u_3$ and of $\triangle v_1v_2v_3$
appears in any diametral pair.
Also note that \figurename~\ref{fig:examples2}(d) consists of \emph{convex} holes only.
We think that any complicated construction can be ``convexified''
in a similar fashion.
This would suggest that computing the diameter in polygonal domains with convex holes only
might be as difficult as the general case.


\subsection{A proof for \figurename~\ref{fig:examples1}(c): Case (I-I) with 6 shortest paths}

\begin{figure}[h]
\begin{center}
  \includegraphics[width=.85\textwidth]{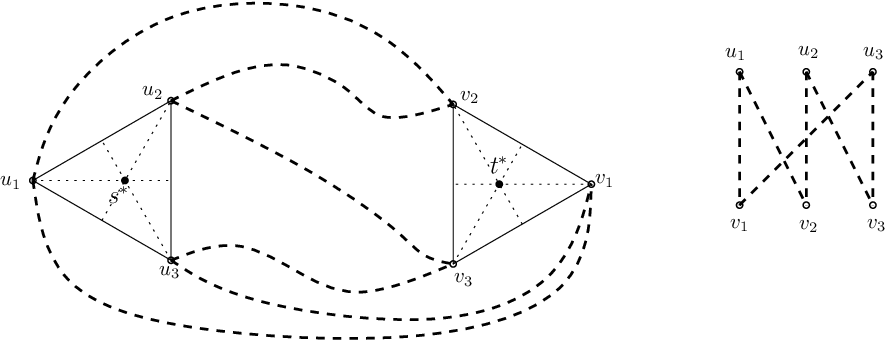}
\end{center}
\caption{\small
 A schematic diagram corresponding to the polygonal domain shown in \figurename~\ref{fig:examples1}(c).
 }
\label{fig:3-3-6}
\end{figure}

\begin{claim}
 In the polygonal domain described in \figurename~\ref{fig:examples1}(c),
 $(s^*, t^*)$ is the unique diametral pair.
\end{claim}
\begin{proof}[Proof of Claim.]
Recall that by construction of the problem instance, the triangles $\triangle u_1u_2u_3$ and $\triangle v_1v_2v_3$ are regular
and $\dist(u_1,v_1)=\dist(u_1, v_2)= \dist(u_2, v_2)=\dist(u_2,v_3)$
$=\dist(u_3,v_3)=\dist(u_3,v_1) = L$, for some arbitrarily large value $L > 0$.
Also, $s^*$ and $t^*$ are the centers of $\triangle u_1u_2u_3$ and $\triangle v_1v_2v_3$,
respectively.

We assume that both triangles $\triangle u_1u_2u_3$ and $\triangle v_1v_2v_3$
are inscribed in a unit circle (and thus $\dist(s^*, t^*) = 2 + L$).
For any point $s$ on any shortest path between $u_i$ and $v_j$,
it is easy to see that $\dist(s, t) \leq \sqrt{3} + L < \dist(s^*,t^*)$
for every point $t \in \Poly$.
In particular, no point on those paths cannot contribute to the diameter.

(1) First, observe that $\max_{t\in \triangle v_1v_2v_3} \dist(s^*,
t) = \max_{s\in \triangle u_1u_2u_3} \dist(s, t^*) = \dist(s^*, t^*)$.

(2) For any $s\in \triangle u_1u_2u_3$, its farthest point $t\in \triangle v_1v_2v_3$ is
on the angle bisector of some $v_i$.
Consider any $s\in \triangle u_1u_2u_3$.
Without loss of generality we assume that $\|s-u_1\| \leq \min_i\{\|s -u_i\|\}$.
Both shortest paths to $v_1$ and to $v_2$ from $s$
pass through $u_1$.
We have $\dist(s, v_1) = \dist(s, v_2)$ by construction
and its farthest point $t\in \triangle v_1v_2v_3$ must be in the angle
bisector of $v_3$.
By symmetry, the same property holds when the closest corner from $s^*$
is either $u_2$ or $u_3$.

Conversely, for any $t$, its farthest point $s\in
\triangle u_1u_2u_3$ must be on a bisector of some $u_i$.
In any diametral pair $(s,t)$, we have that $t$ is the farthest point of
$s$ (and vice versa), so both must be on one of the angle bisectors.

(3) If $(s,t)$ is a diametral pair, then $s \in \seg{u_is^*}$ and
$t\in \seg{v_jt^*}$, for some $i$ and $j$.
Suppose that $s$ lies on the bisector of $u_1$ but not in between $u_1$ and $s^*$.
We then have $\|s-u_2\| = \|s-u_3\| < \|s-u_1\|$ and
$\dist(s,v_1)= \dist(s,v_2)=\dist(s,v_3) = \|s-u_2\|+L$ by construction.
This implies that $t^*$ is the farthest point of such $s$.
Since $\|s-u_2\| < 1$ and $\dist(s,t^*)<2+L$, $(s, t^*)$ is not a diametral pair.

(4) Now, pick any point $s\in \seg{u_1s^*}$ with $s\neq s^*$.
Suppose that $t \in \triangle v_1v_2v_3$ is the farthest point from $s$.
We know that $t \in \seg{v_3t^*}$ by above discussions.
In this case, we have four shortest paths between $s$ and $t$
through $(u_1,v_1)$, $(u_1,v_2)$, $(u_2,v_3)$, and $(u_3,v_3)$;
the other two are strictly longer unless $s = s^*$.
By Theorem~\ref{theorem:charact}, such $s \in \seg{u_1s^*}$ with $s\neq s^*$
and its farthest point $t$ cannot form a maximal pair.
By symmetry, the other cases where $s\in\seg{u_is^*}$ can be handled.

Hence, $(s^*,t^*)$ is a unique diametral pair and the geodesic diameter is $2+L$.
\end{proof}

\subsection{Diametral pair of Case (I-I) with exactly 5 shortest paths}
Here, we present a polygonal domain
in which the diameter is determined by two interior points
and exactly five shortest paths between them.
This proves the tightness of Case \textbf{(I-I)} in Theorem~\ref{theorem:charact}.

\begin{figure}[h]
\begin{center}
  \includegraphics[width=.61\textwidth]{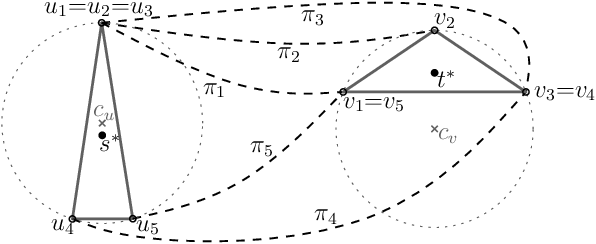}
\end{center}
\caption{\small
 A schematic diagram of a polygonal domain in which $|V_{s^*}|=|V_{t^*}|=3$ and
 $|\Pi(s^*,t^*)| = 5$.
 }
\label{fig:3-3-5}
\end{figure}


\figurename~\ref{fig:3-3-5} shows a schematic description of a polygonal domain $\Poly$.
We assume that only the position of the vertices $u_i$ and the $v_i$ are geometrically precise. We construct the problem instance such that we have $u_1=u_2=u_3$, $v_1=v_5$, and $v_3=v_4$, and
the convex hulls of the $u_i$ and of the $v_i$ form isosceles triangles $\triangle_u$ and $\triangle_v$.
Each of $\triangle_u$ and $\triangle_v$ is inscribed in a unit circle centered at $c_u$ and $c_v$. Moreover, the bases of both triangles are horizontal and the angles opposite to the bases are
$18^\circ$ and $112^\circ$, respectively.
Note that the side lengths of the triangles $\triangle_u$ and $\triangle_v$ are as follows:
$\|u_1 - u_4\| = 1.97537\cdots$ and $\|u_4 -u_5\|= 0.61803\cdots$;
$\|v_2 - v_1\| = 1.11833\cdots$ and $\|v_1 - v_3\| = 1.85436\cdots$.

In this configuration, we set the constants as follows:
letting $L:=\dist(u_1,v_1) = \dist(u_3,v_3)$ be some sufficiently large number,
we set $\dist(u_2,v_2) = L + 0.5$ and $\dist(u_4,v_4) = \dist(u_5,v_5) = L + 0.2$.
Note that this configuration can be realized with four obstacles in a similar way
as \figurename~\ref{fig:examples1}(c).

Since we have fixed all necessary parameters,
we have a fully explicit description of the $\plf_{u_i,v_i}$.
Due to the difficulty of finding an exact analytical solution,
we used numerical methods to solve the system of equations
$\plf_{u_1,v_1}(s,t) = \cdots = \plf_{u_5, v_5}(s,t)$.
We have found that there is a unique solution $(s^*, t^*)$
such that $s^*\in \triangle_u$ and $t^*\in\triangle_v$; we obtained
$s^* = c_u + (0, -0.102795\cdots)$, $t^* = c_v + (0, 0.555361\cdots)$ and
$\dist(s^*,t^*) = 2.047433734\cdots +L$.
(See \figurename~\ref{fig:3-3-5}.)

We first checked that $(s^*, t^*)$ is a maximal pair based on the following lemma,
which can be shown using elementary linear algebra together with the convexity of the path-length functions.
\begin{lemma}
 Suppose that $(s,t)$ is a solution to the system $\plf_{u_1,v_1}(s,t) = \cdots = \plf_{u_5, v_5}(s,t)$.
 If any four of the five gradients $\nabla \plf_{u_i,v_i}$ at $(s,t)$
 are linearly independent (as vectors in a $4$-dimensional space) and
 one of them is represented as a linear combination of
 the other four with all ``negative'' coefficients,
 then $(s,t)$ is a local maximum of the pointwise minimum of the five functions $\plf_{u_i,v_i}$.
 \hfill\copy\ProofSym
\end{lemma}

Next, to see that $(s^*,t^*)$ is a diametral pair,
we have run our algorithm for each of Cases \textbf{(B-B)}, \textbf{(B-I)}, and \textbf{(I-I)};
as a result, there are $44$ candidate pairs, including $(s^*, t^*)$, falling into those cases
among which at most $11$ are maximal and only $(s^*, t^*)$ is diametral.
Note that the pair $(s^*, t^*)$ is the only candidate pair of Case \textbf{(I-I)}.
Also, observe that any point on the shortest path between $u_i$ and $v_i$
cannot belong to a diametral pair.
This implies that none of the $u_i$ and the $v_i$ belongs to a diametral pair. In particular, we have that none of the Cases \textbf{(V-*)} can happen.
In addition, we also sampled about 350,000 points uniformly from each of $\triangle_u$ and $\triangle_v$,
and evaluated the geodesic distances of the 350,000$^2$ pairs.

Note that one can modify the construction to have $|V_{s^*}|=|V_{t^*}|=|\Pi(s^*,t^*)|=5$.
For the purpose, we can split $u_1, u_2, u_3$ into three close corners
(analogously for corners, $v_1, v_5$ and $v_3, v_4$).
The splitting process should preserve the differences between the
distances $\dist(u_i,v_i)$ for all $i=1,\ldots, 5$ (and increase other distances).
We have tested such an example in the same way as above
and concluded that a solution equalizing the five path-length functions
is indeed a diametral pair.

\end{document}